\documentclass[a4paper]{article}

\usepackage{dx-2013}
\usepackage[T1]{fontenc}
\usepackage{ae,aecompl}
\usepackage{amsfonts}
\usepackage{amsmath}
\usepackage{amssymb}
\usepackage{amsthm}
\usepackage{times}
\usepackage{subfigure}
\usepackage{graphicx}
\usepackage{tikz}
\usepackage{url}
\usepackage[draft]{fixme}
\usepackage{algorithmic}
\usepackage{algorithm}
\usepackage[normalem]{ulem}

\usetikzlibrary{automata,petri,positioning}
\usetikzlibrary{shapes,snakes,arrows, automata,shadows,calc}
\tikzstyle{state}=[circle,draw,thick,inner sep=1.5mm]

\title{Distributed Analysis for Diagnosability \\ in Concurrent Systems \thanks{This work has been supported by the European Union Sh Framework Programme under grant agreement no. 295261 (MEALS).}}
\author
{
Hern\'an Ponce de Le\'on$^1$ \and Gonzalo Bonigo$^2$ \and Laura Brand\'an Briones$^{2,3}$ \\
$^1$INRIA and LSV, \'Ecole Normale Sup\'erieure de Cachan and CNRS, France \\
$^2$Fa.M.A.F. - Universidad Nacional de C\'ordoba, Argentina \\
$^3$CONICET \\
e-mail: ponce@lsv.ens-cachan.fr, bonigo@famaf.unc.edu.ar, lbrandan@famaf.unc.edu.ar
}

\newtheorem{assumption}{\normalfont \textbf{Assumption}}

\newcommand{\diag}[1]{\textbf{diag}(#1)}

\newcommand{\imp}{\Rightarrow}
\newcommand{\traces}[1]{\textrm{traces}(#1)}

\newcommand{\trace}[1]{\textrm{trace}(#1)}
\newcommand{\camino}[1]{\textrm{path}(#1)}
\newcommand{\paths}[1]{\textrm{paths}(#1)}
\newcommand{\proj}[2]{P_{#1}({#2})}

\newcommand{\bproof}{$$\begin{array}{ll}}
\newcommand{\eproof}{\end{array}$$}

\newtheorem{definition}{Definition}
\newtheorem{example}{Example}
\newtheorem{proposition}{Proposition}
\newtheorem{theorem}{Theorem}

\begin{document}

\maketitle

\begin{abstract}
Complex systems often exhibit unexpected faults that are difficult to handle. Such systems are desirable to be diagnosable, i.e. faults can be automatically detected as they occur (or shortly afterwards), enabling the system to handle the fault or recover. A system is diagnosable if it is possible to detect every fault, in a finite time after they occurred, by only observing the available information from the system. Complex systems are usually built from simpler components running concurrently. We study how to infer the diagnosability property of a complex system (distributed and with multiple faults) from a parallelized analysis of the diagnosability of each of its components synchronizing  with fault free versions of the others. %
%We present a method where each component can analyzed the diagnosability of the all system in parallel with the other ones.
In this paper we make the following contributions: (1) we address the diagnosability problem of concurrent systems with arbitrary faults occurring freely in each component. (2) We distribute the diagnosability analysis and illustrate our approach with examples. Moreover, (3) we present a prototype tool that implements our techniques showing promising results.
\end{abstract}

\section{Introduction}

As systems become larger, their behavior becomes more complex. Several things may go wrong, resulting in faults occurring. It is then crucially important to design our systems in a way that we can detect or recover from such faults when they occur. A system is diagnosable when its design allows the detection of faults, for instance a system that has sensors specially dedicated to detect them. Sometimes the detection of faults is more involved and the diagnosability property is harder to establish, specially in systems with several components.

A sound software engineering rule for building complex systems is to divide the whole system in smaller and simpler components, each solving a specific task. Moreover, usually they are built by different groups of people and may be in different places. This means that, in general, complex systems are actually collections of simpler components running in parallel. 

In order to model such systems and formally prove results, there are several formalisms like Finite State Machines (FSMs)~\cite{Sampath1995,twin}, Petri Nets~\cite{GencL03,agnesPetri} and Labeled Transition Systems (LTSs)~\cite{laura,agnes,gonza}. In this paper, we model each component by a LTS, so the whole system is a collection of LTSs synchronizing in all their shared observable actions (see Section~\ref{sec:DDES}). 

In the diagnosability analysis of distributed systems it is usually assumed that a fault can occur in exactly one of the different components. We relax this assumption allowing the same fault to occur in several components. 

Also, the diagnosability analysis is usually iterative (i.e., sequential): the information from local diagnosers is combined until a global verdict is reached. We propose a method to distribute this analysis.

Finally, we developed a tool that implements all our research. The DADDY tool (Distributed Analysis for Distributed Discrete sYstems)~\cite{daddy} is a prototype based on the results presented in~\cite{gonza} and this paper. The tool does not only implement the method we presents but also the classic one allowing us to compare both approaches. We present a comparative analysis of their performance obtained from the experimental running of several examples. 

\paragraph{Related Work} Diagnosability was initially developed in~\cite{Sampath1995} under the setting of discrete event systems. In that paper, necessary and sufficient conditions for testing diagnosability are given. In order to test diagnosability, a special diagnoser is computed, whose complexity of construction is shown to be exponential in the number of states of the original system, and double exponential in the number of faults. Later, in~\cite{twin}, an improvement of this algorithm is presented, where the so-called twin plant method is introduced and shown to have polynomial complexity in the number of states and faults. Afterwards, in~\cite{Schumann07scalablediagnosability}, an improvement to the twin plant method is presented where the system is reduced before building the twin plant.

None of the methods presented there (i.e., \cite{Sampath1995,twin}) consider the problem when the system is composed of components working in parallel. An approach to this consideration is addressed in~\cite{Schumann07scalablediagnosability,Debouk_acoordinated,distributeddiag,SchumannH08} where the diagnosability problem is performed by either local diagnosers or twin plants communicating with each other, directly or through a coordinator, and by that means pooling together the observations. \cite{YeDagueValid} shows that when considering only local observations, diagnosability becomes undecidable when the communication between component is unobservable. An algorithm is proposed to check a sufficient but not necessary condition of diagnosability. However, their results are based in the assumption that a fault can only occur in one of the components, an assumption that can not always be made.

Several mechanisms such as interleaving, shared variables and handshaking have been described in~\cite{0020348} to provide operational models for distributed systems. In the handshaking method, the communication is made by the synchronization on actions or events. These actions must be specified a priori in the model, so the different components can be synchronized at execution time. In~\cite{gonza} the authors study how different kinds of synchronizations (via all the shared actions, some of them or none) impact in the diagnosis analysis. 

\paragraph{Motivation} Suppose different groups of people are commanded to build different components of a system. Even if each component is diagnosable, it is not always the case that the resulting system has such property\footnote{See for example $C,D$ and $C \times D$ in Figures~\ref{fig:lts} and~\ref{fig:prod}.}. In~\cite{gonza} the authors show that with different kinds of synchronizations, the diagnosability of the global system can not be inferred directly from the diagnosability of each component. 

We propose a framework where each component only shares with the rest a fault free version of its own, maybe the specification of its ideal behavior. Then, each component should not only be diagnosable, but also its interaction with the fault free version of the others, i.e. its synchronous product with fault free version of the other components. Therefore, our diagnosability analysis can be distributed.

\paragraph{Paper organization} Section~\ref{sec:DDES} presents the formal model that we use for modeling each component, the parallel composition between them and the notion of diagnosability. In Section~\ref{sec:DDA}, we develop our analysis method, showing how the diagnosability of each component synchronizing with fault free versions of the other components influences the diagnosability property of the overall system. Section~\ref{sec:daddy} presents our tool DADDY and some experimental results. We conclude and discuss about future work in Section~\ref{sec:CFW}.

%%%%%%%%%%%%%%%%
% section sec:DDES
%%%%%%%%%%%%%%%%
\section{Diagnosability Analysis}\label{sec:DDES}%in Discrete Event Systems}\label{sec:DDES}

\subsection{Model of the system}

We consider a distributed system composed of two autonomous components $G_1, G_2$ that communicate with each other by all their shared observable actions. The local model of a component is defined as a Labeled Transition System.

\begin{definition}\label{LTS}
  A Labeled Transition System (LTS) is a tuple $G = (Q, \Sigma, \delta, q_0)$ where
  \begin{itemize}
    \item $Q$ is a finite set of states,
    \item $\Sigma$ is a finite set of actions,
    \item $\delta$ is a partial transition function, and 
    \item $q_0$ the initial state, with $q_0 \in Q$.
  \end{itemize}
%  In this paper we consider LTS with a finite set of states\footnote{The results of this papers also hold for Finite State Machines.}.
\end{definition}

As usual in diagnosability analysis, some of the actions of $\Sigma$ are observable while the rest are unobservable. Thus the set of actions $\Sigma$ is partitioned as $\Sigma = \Sigma_o \uplus \Sigma_{uo}$ where $\Sigma_o$ represents the observable actions and $\Sigma_{uo}$ the unobservable ones. 

The faults to diagnose are considered unobservable, i.e. $\Sigma_F \subseteq \Sigma_{uo}$, as faults that are observable can be easily diagnosable.

As usual in diagnosability analysis, we made the following assumptions about our systems.

\begin{assumption}
  We only consider (live) systems where there is a transition defined at each state, i.e. the system cannot reach a point at which no action is possible. 
\end{assumption}

\begin{assumption}
  The system does not contain cycles of unobservable actions.
\end{assumption}

Note that, these assumptions together assure that all our systems are free of observation starvation.

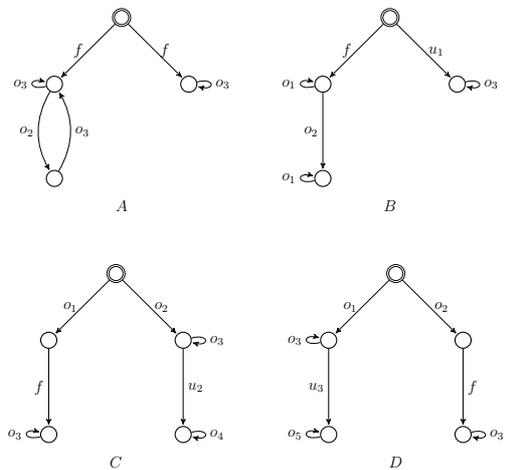
\begin{figure}[h]\label{Figure1}
  \centering
  \subfigure{\scalebox{.5}{\begin{tikzpicture}[->,>=stealth',shorten >=1pt,auto,node distance=2.5cm,semithick]
\tikzstyle{every state}=[draw]

  \node[accepting,state] 		(A)                    {};
  \node[state]         	(B) [below left of=A] {};
  \node[state]         	(C) [below right of=A] {};
  \node[state]         	(D) [below of=B] {};

  \path	(A)	edge [left]		node {\Large $f$}		(B)
            	(A)	edge	 [right]	node {\Large $f$}		(C)
	         (B)	edge	 [bend right, left]		node {\Large $o_2$}		(D)
		(B)	edge	 [loop left]	node {\Large $o_3$}		(B)
		(D)	edge	 [bend right,right]	node {\Large $o_3$}		(B)
		(C)	edge	 [loop right]	node {\Large $o_3$}		(C);
		
\node[] (name) at (0,-5) {\Large $A$};					
\node[] (null) at (0,-6) {};					

\end{tikzpicture}}}
  \hspace{3.5mm}
  \subfigure{\scalebox{.5}{\begin{tikzpicture}[->,>=stealth',shorten >=1pt,auto,node distance=2.5cm,semithick]
\tikzstyle{every state}=[draw]

  \node[accepting,state] 		(A)                    {};
  \node[state]         	(B) [below left of=A] {};
  \node[state]         	(C) [below right of=A] {};
  \node[state]         	(D) [below of=B] {};

%\tikzset{font=\bfseries}
  \path	(A)	edge [left]		node {\Large{\bf $f$}}		(B)
  		(B)	edge	 [loop left]	node {\Large $o_1$}		(B)
            	(A)	edge	 [right]	node {\Large $u_1$}		(C)
	         (B)	edge	 [left]		node {\Large $o_2$}		(D)
		(D)	edge	 [loop left]	node {\Large $o_1$}		(D)
		(C)	edge	 [loop right]	node {\Large $o_3$}		(C);
		
\node[] (name) at (0,-5) {\Large $B$};							
\node[] (null) at (0,-6) {};					

\end{tikzpicture}}}
  \subfigure{\scalebox{.5}{\begin{tikzpicture}[->,>=stealth',shorten >=1pt,auto,node distance=2.5cm,semithick]
\tikzstyle{every state}=[draw]

  \node[accepting,state] 		(A)                    {};
  \node[state]         	(B) [below left of=A] {};
  \node[state]         	(C) [below right of=A] {};
  \node[state]         	(D) [below of=B] {};
  \node[state]         	(E) [below of=C] {};

  \path	(A)	edge [left]		node {\Large $o_1$}		(B)
            	(A)	edge	 [right]		node {\Large $o_2$}		(C)
	         (C)	edge	 [loop right]	node {\Large $o_3$}		(C)
		(B)	edge	 [left]		node {\Large $f$}		(D)
		(D)	edge	 [loop left]	node {\Large $o_3$}		(D)	
		(C)	edge	 [right]		node {\Large $u_2$}		(E)
		(E)	edge	 [loop right]	node {\Large $o_4$}		(E);

\node[] (name) at (0,-5) {\Large $C$};			

\end{tikzpicture}}}
  \hspace{5mm}
  \subfigure{\scalebox{.5}{\begin{tikzpicture}[->,>=stealth',shorten >=1pt,auto,node distance=2.5cm,semithick]
\tikzstyle{every state}=[draw]

  \node[accepting,state] 		(A)                    {};
  \node[state]         	(B) [below right of=A] {};
  \node[state]         	(C) [below left of=A] {};
  \node[state]         	(D) [below of=B] {};
  \node[state]         	(E) [below of=C] {};

  \path	(A)	edge [right]		node {\Large $o_2$}		(B)
            	(A)	edge	 [left]		node {\Large $o_1$}		(C)
	         (C)	edge	 [loop left]	node {\Large $o_3$}		(C)
		(B)	edge	 [right]		node {\Large $f$}		(D)
		(D)	edge	 [loop right]	node {\Large $o_3$}		(D)	
		(C)	edge	 [left]		node {\Large $u_3$}		(E)
		(E)	edge	 [loop left]	node {\Large $o_5$}		(E);

\node[] (name) at (0,-5) {\Large $D$};			

\end{tikzpicture}}}
  \caption{Specification of four components modeled by LTSs}
   \label{fig:lts}
\end{figure}

Figure~\ref{fig:lts} shows four components modeled by the LTSs $A,B,C$ and $D$ where $o_1,o_2,o_3,o_4,o_5 \in \Sigma_o$ and $u_1,u_2,u_3 \in \Sigma_{uo}$. The special action $f \in \Sigma_F$ is the fault to be diagnosable.

A path from state $q_i$ to state $q_j$ in $G$ is a sequence $q_i \cdot a_i \cdot q_{i+1} \dots a_{j-1} \cdot q_j$ such that $(q_k, a_k, q_{k+1}) \in \delta$ for $i \leq k \leq j-1$. The set of paths in $G$ is denoted by $\paths G$.

The trace associated with any given path consists of its sequence of actions (i.e., for a path $\rho = q_0 \cdot a_0 \cdot q_1 \dots a_{n-1} \cdot q_n$ we have $\trace \rho = a_0 \cdot a_1 \dots a_n$). Given a trace, $\sigma = a_0 \cdot a_1 \dots a_n$, we denote as $f \in \sigma$ when there exists $i$ such that $f = a_i$. As our systems are live, we only consider infinity traces where the infinite repetition of an actions $a$ is denoted by $\widehat{a}$. The set of all traces starting in $q_0$ is denoted by $\traces G$. As we consider nondeterministic systems, the same trace can belong to several paths. The set of possible paths of a trace $\sigma$ in $G$ are: $\camino \sigma = \{ \rho \in \paths G \mid \trace \rho = \sigma \}$. 

The observation of a trace is given by the following definition.

\begin{definition} \label{def:obs}
Let $\sigma \in \Sigma^*$, then
\begin{eqnarray*}
  obs(\sigma) & = &
  \left \{
  \begin{array}{l l l}
    \epsilon & \text{if } \sigma = \epsilon \\
    a\!\cdot\! obs(\sigma') & \text{if } \sigma = a\!\cdot\!\sigma' \land a\!\in\!\Sigma_o \\
    obs(\sigma') & \text{if } \sigma = a\!\cdot\!\sigma' \land a\!\not \in\!\Sigma_o \\
  \end{array}
  \right.
\end{eqnarray*}
\end{definition}

The communication between two components is given by their synchronous product where the synchronizing actions are all the shared observable ones.

\begin{definition} \label{def_synchronization}
  Given two local components $G_1 = (Q^1, \Sigma^1, \delta^1, q^1_0)$ and $G_2 = (Q^2, \Sigma^2, \delta^2, q^2_0)$, the behavior of the global system is given by their synchronous product $G_1 \times G_2 = (Q^1 \times Q^2, \Sigma^1 \cup \Sigma^2, \delta^{1 \times 2}, (q^1_0,q^2_0))$ where $\delta^{1 \times 2}$ is defined as follows
\begin{eqnarray*}
  \delta^{1 \times 2}((q^1_i,q^2_j),a) \hspace{-3mm} & =\hspace{-3.5mm} &
  \left \{\hspace{-2.5mm}
  \begin{array}{l l l}
    (\delta^1(q^1_i\!,\!a),\delta^2(q^2_j\!,\!a)) & \text{if}\ a\!\in\!\Sigma^1_o \cap \Sigma^2_o \\
    (\delta^1(q^1_i,a), q^2_j) & \text{if}\ a\!\in\!\Sigma^1 \wedge a\!\not \in\!\Sigma^2 \\
    (q^1_i, \delta^2(q^2_j,a)) & \text{if}\ a\!\in\!\Sigma^2 \wedge a\!\not \in\!\Sigma^1 \\
  \end{array}
  \right.
\end{eqnarray*}
\end{definition}

Given a path in the global system, we can project it to a single component.

\begin{definition} \label{def:projrun}
  Let $\rho\!\in\!\paths{G_1\!\times\!G_2}$, its projection in $G_i$ is
\begin{eqnarray*}
  P_i((q^1, q^2)) & = & q^i \\
  P_i((q^1, q^2)\!\cdot\!a\!\cdot\!\rho') & = & 
  \left \{
  \begin{array}{lcl}
  q^i\!\cdot\!a\!\cdot\!P_i(\rho') & &\text{if } \exists\ \delta^i(q^i, a) \\
  P_i(\rho') & & \text{otherwise}
  \end{array}
  \right.
\end{eqnarray*}
\end{definition}

For a trace in the global system, we define the projections to know which actions belong to a certain component.

\begin{definition} \label{def:projtrace}
 Let $\sigma$ be a trace in $\traces{G_1 \times G_2}$, $\sigma'$ is its projection in $G_i$, denoted $\proj{i}{\sigma}=\sigma'$, iff
\begin{center}
$\exists\ \rho\!\in\!\camino{\sigma} : \trace{P_i(\rho)}\!=\!\sigma'$
\end{center}
\end{definition}

\begin{example}
  Let $\sigma = o_1 f  o_3 u_3  \widehat{o}_5$ be a trace in $\traces{C \times D}$ from Figure~\ref{fig:prod}, its projection in component $C$ is given by $\proj{C}{\sigma} = o_1 f o_3$ and its projection in component $D$ is given by $\proj{D}{\sigma} = o_1 o_3 u_3 \widehat{o}_5$. These projections are traces of the corresponding components $C$ and $D$ from Figure~\ref{fig:lts}. Note that projections of an infinite trace from the global system can be finite in one of the components.
\hspace{1cc}

\begin{figure}[h]
  \centering
  \subfigure{\scalebox{.34}{\begin{tikzpicture}[->,>=stealth',shorten >=1pt,auto,node distance=2.5cm,semithick]
\tikzstyle{every state}=[draw]

  \node[accepting,state] 		(0)                    {};
  \node[state]         	(5) [below left of=0] {};
  \node[state]         	(9) [below right of=0] {};
  
  \node[state]         	(4) [below left of=5] {};
  \node[state]         	(2) [above left of=4] {};
  \node[state]         	(1) [below of=4] {};
  \node[state]         	(8) [below of=1] {};
    
  \node[state]         	(7) [below right of=9] {};
  \node[state]         	(3) [above right of=7] {};
  \node[]         	(null) [below of=7] {};
  \node[state]         	(6) [below of=null] {};

  \path	(0)	edge [above]		node {\Huge $f$}		(2)
  		(0)	edge	 [left]		node {\Huge $f$}		(5)
		(0)	edge	 [right]	node {\Huge $f$}		(9)
            	(0)	edge	 [above]		node {\Huge $u_1$}		(3)

  		(2)	edge	 [left]		node {\Huge $f$}		(4)
		(5)	edge	 [right]	node {\Huge $f$}		(4)
            	(9)	edge	 [left]		node {\Huge $u_1$}		(7)
  		(3)	edge	 [right]		node {\Huge $f$}		(7)

  		(2)	edge	 [left]		node {\Huge $f$}		(8)
		(5)	edge	 [right]	node {\Huge $u_1$}		(6)
            	(9)	edge	 [left]		node {\Huge $f$}		(8)
  		(3)	edge	 [right]		node {\Huge $f$}		(6)
		
  		(4)	edge	 [right]		node {\Huge $o_2$}		(1)

  		(2)	edge	 [loop left]		node {\Huge $o_1$}		(2)
		(4)	edge	 [loop right]	node {\Huge $o_1$}		(4)
            	(7)	edge	 [loop left]		node {\Huge $o_3$}		(7)
  		(1)	edge	 [loop right]		node {\Huge $o_1$}		(1)
  		(8)	edge	 [loop left]		node {\Huge $o_1$}		(8)
		(6)	edge	 [loop right]	node {\Huge $o_3$}		(6);
				
\node[] (name) at (0,-9) {\Huge $A \times B$};			

\end{tikzpicture}}}
  \subfigure{\scalebox{.34}{\begin{tikzpicture}[->,>=stealth',shorten >=1pt,auto,node distance=2.5cm,semithick]
\tikzstyle{every state}=[draw]

  \node[accepting,state] 		(0)                    {};
  \node			(null) [below of=0] {};
  \node[state]         	(4) [left of=null] {};
  \node[state]         	(1) [below left of=4] {};
  \node[state]         	(5) [below right of=4] {};
  \node[state]         	(6) [below right of=1] {};
  \node[state]         	(2) [right of=null] {};
  \node[state]         	(7) [below left of=2] {};
  \node[state]         	(8) [below right of=2] {};
  \node[state]         	(9) [below right of=7] {};

  \path	(0)	edge [left]		node {\Huge $o_1$}		(4)
  		(4)	edge	 [left]		node {\Huge $f$}		(1)
		(4)	edge	 [right]	node {\Huge $u_3$}		(5)
            	(1)	edge	 [left]		node {\Huge $u_3$}		(6)
	         (5)	edge	 [right]	node {\Huge $f$}		(6)

		(0)	edge	 [right]	node {\Huge $o_2$}		(2)
  		(2)	edge	 [left]		node {\Huge $u_2$}		(7)
		(2)	edge	 [right]	node {\Huge $f$}		(8)
            	(7)	edge	 [left]		node {\Huge $f$}		(9)
	         (8)	edge	 [right]	node {\Huge $u_2$}		(9)

		(1)	edge	 [loop left]	node {\Huge $o_3$}		(1)
		(5)	edge	 [loop left]node {\Huge $o_5$}		(5)
		(6)	edge	 [loop below]	node {\Huge $o_5$}		(6)

		(7)	edge	 [loop right]	node {\Huge $o_4$}		(7)
		(8)	edge	 [loop right]	node {\Huge $o_3$}		(8)
		(9)	edge	 [loop below]	node {\Huge $o_4$}		(9);
				
\node[] (name) at (0,-9) {\Huge $C \times D$};			

\end{tikzpicture}}}
  \caption{Synchronous product of components $A, B$ and $C,D$}
  \label{fig:prod}
\end{figure}
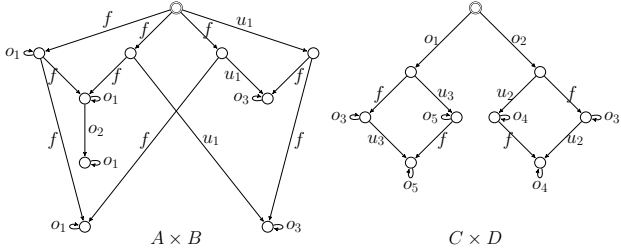
\end{example}

As the projection operator only erases actions in a trace, it is easy to see that every fault belonging to a projection of such a trace, also belongs to the trace in the global system as it is shown by the following result.

\begin{proposition} \label{proj_inh_fault}
For every trace $\sigma$ in $\traces{G_1 \times G_2}$ with $\proj{i}{\sigma} = \sigma_i$, we have
$$\text{if }\ \  f \in \sigma_i \ \ \text{ then }\ \  f \in \sigma$$
\end{proposition}

When two components synchronize in all their shared actions, if two traces of the global system have the same observability and we project them to the same component, the resulting projections will also have the same observability. This result is captured by Proposition~\ref{proj_obs}.

\begin{proposition} \label{proj_obs}
Given two traces $\sigma$ and $\alpha$ in $\traces{G_1 \times G_2}$ with $\proj{i}{\sigma} = \sigma_i \text{ and } \proj{i}{\alpha} = \alpha_i$, we have
$$\text{if}\ \ obs(\sigma) = obs(\alpha)\ \ \text{ then }\ \ obs(\sigma_i) = obs(\alpha_i)$$
\end{proposition}

This result is proved by double induction in the structure of $\sigma$ and $\alpha$. We analyze several cases depending on the existence of the projections. One of the most critical cases is when $\sigma = a\!\cdot\!\sigma', \alpha = b\!\cdot\!\alpha', a \in \Sigma^1_o \cap \Sigma^2_o, \text{but } b \not\in \Sigma^1_o \cap \Sigma^2_o$ as it has several particular sub-cases. Note that this result only holds when the synchronization is done in all the set of shared actions.

%%%%%%%%%%%%%%%%
% section sec:DC
%%%%%%%%%%%%%%%%
\subsection{Diagnosability condition}

We present now the notion of diagnosability. Informally, a fault $f \in \Sigma_F$ is diagnosable if it is possible to detect, within a finite delay, occurrences of such a fault using the record of observed actions. In other words, a fault is not diagnosable if there exist two infinite paths from the initial state with the same infinite sequence of observable actions but only one of them contains a fault.

\begin{definition}
  Let $f$ be a fault in $\Sigma_F$, $f$ is diagnosable in $G$ iff
  $$\forall \sigma, \alpha \in \traces G : \text{if } obs(\sigma) = obs(\alpha)$$ $$\text{ and } f \in \sigma \text{ then } f \in \alpha$$
  
  The system $G$ is diagnosable, denoted by \diag G, if and only if every fault $f \in \Sigma_F$ is diagnosable.
\end{definition}

The previous definition introduced in~\cite{laura} is a reformulation of the one presented in~\cite{Sampath1995}.

\begin{example} \label{ex:prod}
 Let consider the components $A$ and $B$ from Figure~\ref{fig:lts}. The only pair of traces in $A$ with the same observability are of the form $f \widehat{o}_3$ (one for each branch from the initial state), as both traces contain the fault $f$, system A is diagnosable. In the case of $B$, each observable trace corresponds to a unique path, therefore $B$ is diagnosable. 

Now, consider system $A \times B$ from Figure~\ref{fig:prod}, we can see that every trace contains a fault, therefore $A \times B$ is diagnosable. On the contrary, in system $C \times D$ we have two traces, $o_2 u_2 \widehat{o}_4$ and $o_2 f u_2 \widehat{o}_4$ that have the same observability, but one of them contains a fault and the other does not, therefore the system $C \times D$ is not diagnosable.
\end{example}

%%%%%%%%%%%%%%%%
% section sec:DDA
%%%%%%%%%%%%%%%%
\section{Distributing the diagnosability analysis}\label{sec:DDA}

The notion of diagnosability is introduced in~\cite{Sampath1995} assuming a centralized architecture of the system. In order to check the diagnosability property in distributed systems, the synchronous product of components is computed and such a product is given as an input to an algorithm that tests its diagnosability (usually based on the twin plant method). The size of such a product grows exponentially with respect of the size of the components, resulting in an inefficient algorithm. When dealing with real applications, such as telecommunication networks or power distribution networks, the centralized approach is clearly unrealistic because of the size of those applications. Moreover, this approach does not exploit the fact that such systems are distributed.

In~\cite{Schumann07scalablediagnosability,distributeddiag} the authors distribute the search for non-distinguishable behaviors based on local verifiers and local twin plants. The local information is propagated until a verdict is made or, in the worst case, the global system is built. Their result is based on the assumption that a fault can occur in exactly one component.

In this section we present a method that allows to decide the diagnosability of a distributed system in terms of the diagnosability of each faulty component synchronizing with fault free versions of the remaining ones. Basically, we compose each component with a fault free version of the other components and analyze their diagnosability in parallel. To the best of our knowledge, it is the first method that allows to parallelize the diagnosability analysis.

%\cap (Q' \times \Sigma \times Q)$}
\begin{algorithm} \label{algo}
  \caption{}
  \begin{algorithmic}[1] \label{algo}
    \REQUIRE A LTS $G = (Q, \Sigma, \delta, q_0)$
    \ENSURE An $f$-fault free version of $G$
      \STATE $Q' := \{ q_0 \}$ , $\delta' := \emptyset$ , $Q := Q \setminus \{ q_0 \}$
      \WHILE{$\exists (q'\!,\!x,\!q)\!:\!q'\!\!\in\!Q'\!\wedge (q'\!,\!x,\!q)\!\in\!\delta \wedge (q',\!x,\!q)\!\not\in\!\delta'$}
        \IF{$x \not = f$}
          \STATE $Q' := Q' \cup \{ q \}$
          \STATE $\delta' := \delta' \cup (q',x,q)$
        \ENDIF
        \STATE $\delta := \delta \setminus (q',x,q)$
      \ENDWHILE
    \RETURN $G^f = (Q',\Sigma,\delta',q_0)$
  \end{algorithmic}
\end{algorithm}

For testing the diagnosability of a fault $f \in \Sigma_F$ in the global system, instead of computing the whole composition, we consider one component and compose it with the fault free versions of the others. These fault free versions may be taken as the specification of each component, when provided, or can be computed by removing the fault $f$ in the component using Algorithm~\ref{algo} and considering such as the correct behavior of the system.

%Figure~\ref{fig:reducedcompo} shows components $A,B,C$ and $D$ after removing their faults.
\vspace{0.3cc}
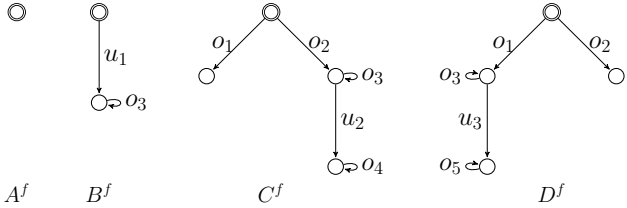
\begin{figure}[h]
  \centering
  \subfigure{\scalebox{.48}{\begin{tikzpicture}[->,>=stealth',shorten >=1pt,auto,node distance=2.5cm,semithick]
\tikzstyle{every state}=[draw]

  \node[accepting,state] 		(A)                    {};
%  \node[state]         	(B) [below left of=A] {};
%  \node[state]         	(C) [below right of=A] {};
%  \node[state]         	(D) [below of=B] {};
%
%  \path	(A)	edge [left]		node {\Huge $f$}		(B)
%            	(A)	edge	 [right]	node {\Huge $f$}		(C)
%	         (B)	edge	 [bend right, left]		node {\Huge $o_2$}		(D)
%		(B)	edge	 [loop left]	node {\Huge $o_3$}		(B)
%		(D)	edge	 [bend right,right]	node {\Huge $o_3$}		(B)
%		(C)	edge	 [loop right]	node {\Huge $o_3$}		(C);
		
\node[] (name) at (0,-5) {\huge $A^f$};					
%\node[] (null) at (0,-6) {};					

\end{tikzpicture}}}
  \hspace{4mm}
  \subfigure{\scalebox{.48}{\begin{tikzpicture}[->,>=stealth',shorten >=1pt,auto,node distance=2.5cm,semithick]
\tikzstyle{every state}=[draw]

  \node[accepting,state] 		(A)                    {};
  %\node[state]         	(B) [below left of=A] {};
  \node[state]         	(C) [below of=A] {};
  %\node[state]         	(D) [below of=B] {};

  \path	%(A)	edge [left]		node {\Huge $f$}		(B)
%  		(B)	edge	 [loop left]	node {\Huge $o_1$}		(B)
            	(A)	edge	 [right]	node {\Huge $u_1$}		(C)
%	         (B)	edge	 [left]		node {\Huge $o_2$}		(D);
%		(D)	edge	 [loop left]	node {\Huge $o_1$}		(D)
		(C)	edge	 [loop right]	node {\Huge $o_3$}		(C);
		
\node[] (name) at (0,-5) {\huge $B^f$};							
%\node[] (null) at (0,-6) {};					

\end{tikzpicture}}}
  \hspace{4mm}
  \subfigure{\scalebox{.48}{\begin{tikzpicture}[->,>=stealth',shorten >=1pt,auto,node distance=2.5cm,semithick]
\tikzstyle{every state}=[draw]

  \node[accepting,state] 		(A)                    {};
  \node[state]         	(B) [below left of=A] {};
  \node[state]         	(C) [below right of=A] {};
  %\node[state]         	(D) [below of=B] {};
  \node[state]         	(E) [below of=C] {};

  \path	(A)	edge [left]		node {\Huge $o_1$}		(B)
            	(A)	edge	 [right]		node {\Huge $o_2$}		(C)
	         (C)	edge	 [loop right]	node {\Huge $o_3$}		(C)
%		(B)	edge	 [left]		node {\Huge $f$}		(D)
%		(D)	edge	 [loop left]	node {\Huge $o_3$}		(D)	
		(C)	edge	 [right]		node {\Huge $u_2$}		(E)
		(E)	edge	 [loop right]	node {\Huge $o_4$}		(E);

\node[] (name) at (0,-5) {\huge $C^f$};			

\end{tikzpicture}}}
  \hspace{4mm}
  \subfigure{\scalebox{.48}{\begin{tikzpicture}[->,>=stealth',shorten >=1pt,auto,node distance=2.5cm,semithick]
\tikzstyle{every state}=[draw]

  \node[accepting,state] 		(A)                    {};
  \node[state]         	(B) [below right of=A] {};
  \node[state]         	(C) [below left of=A] {};
%  \node[state]         	(D) [below of=B] {};
  \node[state]         	(E) [below of=C] {};

  \path	(A)	edge [right]		node {\Huge $o_2$}		(B)
            	(A)	edge	 [left]		node {\Huge $o_1$}		(C)
	         (C)	edge	 [loop left]	node {\Huge $o_3$}		(C)
%		(B)	edge	 [right]		node {\Huge $f$}		(D)
%		(D)	edge	 [loop right]	node {\Huge $o_3$}		(D)	
		(C)	edge	 [left]		node {\Huge $u_3$}		(E)
		(E)	edge	 [loop left]	node {\Huge $o_5$}		(E);

\node[] (name) at (0,-5) {\huge $D^f$};			

\end{tikzpicture}}}
  \caption{Components $A,B,C$ and $D$ after removing their faults}
  \label{fig:reducedcompo}
\end{figure}

If we compose a component $G_i$ with the fault free version of $G_j$, meaning $G^f_j$, clearly the traces of the resulting system are those of $G_i \times G_j$ such that its projections in $G_j$ are fault free.

\begin{proposition} \label{def:redprod}
  Let  $G_i$ and $G_j$ be two LTSs, then $\sigma \in \traces{G_i\!\times\!G^f_j}$ iff 
$$\sigma \in \traces{G_i \times G_j} \wedge \forall \sigma_j : \proj{j}{\sigma} = \sigma_j \imp f \not \in \sigma_j$$
  \label{def:red}
\end{proposition}

Figure~\ref{fig:reducedcompo} shows components $A,B,C$ and $D$ after removing fault $f$ and Figure~\ref{fig:reduced} shows them synchronizing with the faulty components.

\begin{example}\label{ex:reduced}
 Let us consider the systems from Figure~\ref{fig:reduced}. System $A^f \times B$ is trivially diagnosable. In the case of system $A \times B^f$, it is easy to see that the observable traces are of the form $\widehat{o}_3$, but all traces containing $o_3$ also contain $f$ and therefore $A \times B^f$ is also diagnosable. In system $C^f \times D$, traces $\sigma = o_2 u_2 \widehat{o}_4$ and $\alpha = o_2 f u_2 \widehat{o}_4$ have the same observability, but $\alpha$ contains a fault and $\sigma$ does not. So, we can conclude that $C^f \times D$ is not diagnosable.
\end{example}

\begin{figure}
  \centering
  \subfigure{\scalebox{.48}{\begin{tikzpicture}[->,>=stealth',shorten >=1pt,auto,node distance=2.5cm,semithick]
\tikzstyle{every state}=[draw]

  \node[accepting,state] 		(0)                    {};
  \node[state]         	(1) [below of=0] {};

  \path	(0)	edge [left]		node {\Huge $f$}		(1)
		(1)	edge	 [loop right]	node {\Huge $o_1$}		(1);	
		
\node[] (name) at (0,-4) {\huge $A^{f} \times B$};						
\node[] (null) at (0,-5) {};					

\end{tikzpicture}}}
  \hspace{10mm}
  \subfigure{\scalebox{.48}{\begin{tikzpicture}[->,>=stealth',shorten >=1pt,auto,node distance=2.5cm,semithick]
\tikzstyle{every state}=[draw]

  \node[accepting,state] 		(0)                    {};
  \node[state]         	(1) [below of=0] {};
  \node[state]         	(2) [below left of=0] {};
  \node[state]         	(4) [below right of=0] {};
  \node[state]         	(3) [below of=2] {};
  \node[state]         	(5) [below of=4] {};

  \path	(0)	edge [left]		node {\Huge $u_1$}		(1)
  		(0)	edge	 [left]	node {\Huge $f$}			(2)
		(0)	edge	 [right]	node {\Huge $f$}		(4)
            	(2)	edge	 [left]	node {\Huge $u_1$}			(3)
	         (1)	edge	 [left]		node {\Huge $f$}		(3)
		(4)	edge	 [right]	node {\Huge $u_1$}		(5)
	         (1)	edge	 [right]		node {\Huge $f$}		(5)
		(3)	edge	 [loop left]	node {\Huge $o_3$}		(3)
		(5)	edge	 [loop right]	node {\Huge $o_3$}		(5);		
		
\node[] (name) at (0,-5) {\huge $A \times B^{f}$};					
\node[] (null) at (0,-6) {};					

\end{tikzpicture}}}
  \vspace{1.5mm}
  \subfigure{\scalebox{.48}{\begin{tikzpicture}[->,>=stealth',shorten >=1pt,auto,node distance=2.5cm,semithick]
\tikzstyle{every state}=[draw]

  \node[accepting,state] 		(0)                    {};
  \node[state]         	(4) [below right of=0] {};
  \node[state]         	(1) [below left of=4] {};
  \node[state]         	(5) [below right of=4] {};
  \node[state]         	(6) [below right of=1] {};
  \node[state]         	(2) [below left of=0] {};
  \node[state]         	(3) [below of=2] {};

  \path	(0)	edge [right]		node {\Huge $o_2$}		(4)
  		(4)	edge	 [left]		node {\Huge $u_2$}		(1)
		(4)	edge	 [right]	node {\Huge $f$}		(5)
            	(1)	edge	 [left]		node {\Huge $f$}		(6)
	         (5)	edge	 [right]	node {\Huge $u_2$}		(6)
		(0)	edge	 [left]	node {\Huge $o_1$}		(2)
	         (2)	edge	 [left]	node {\Huge $u_3$}		(3)

		(1)	edge	 [loop left]	node {\Huge $o_4$}		(1)
		(5)	edge	 [loop right]node {\Huge $o_3$}	(5)
		(6)	edge	 [loop left]	node {\Huge $o_4$}		(6)
		(3)	edge	 [loop left]	node {\Huge $o_5$}		(3);
		
\node[] (null) at (-5,0) {};	
\node[] (name) at (0,-6.5) {\huge $C^{f} \times D$};			

\end{tikzpicture}}}
  \vspace{1.5mm}
  \subfigure{\scalebox{.48}{\begin{tikzpicture}[->,>=stealth',shorten >=1pt,auto,node distance=2.5cm,semithick]
\tikzstyle{every state}=[draw]

  \node[accepting,state] 		(0)                    {};
  \node[state]         	(5) [below right of=0] {};
  \node[state]         	(6) [below of=5] {};
  \node[state]         	(3) [below left of=0] {};
  \node[state]         	(2) [below left of=3] {};
  \node[state]         	(4) [below right of=3] {};
  \node[state]         	(1) [below right of=2] {};

  \path	(0)	edge [right]		node {\Huge $o_2$}		(5)
  		(5)	edge	 [right]		node {\Huge $u_2$}		(6)
		(0)	edge	 [left]	node {\Huge $o_1$}		(3)
            	(3)	edge	 [left]		node {\Huge $f$}		(2)
	         (2)	edge	 [left]	node {\Huge $u_3$}		(1)
		(3)	edge	 [right]	node {\Huge $u_3$}		(4)
	         (4)	edge	 [right]	node {\Huge $f$}		(1)

		(6)	edge	 [loop right]	node {\Huge $o_4$}		(6)
		(2)	edge	 [loop left]node {\Huge $o_3$}	(2)
		(4)	edge	 [loop right]	node {\Huge $o_5$}		(4)
		(1)	edge	 [loop right]	node {\Huge $o_5$}		(1);

\node[] (null) at (5,0) {};			
\node[] (name) at (0,-6.5) {\huge $C \times D^{f}$};			
	
\end{tikzpicture}}}
  \caption{Composed systems after removing the faults in one of the components}
  \label{fig:reduced}
\end{figure}
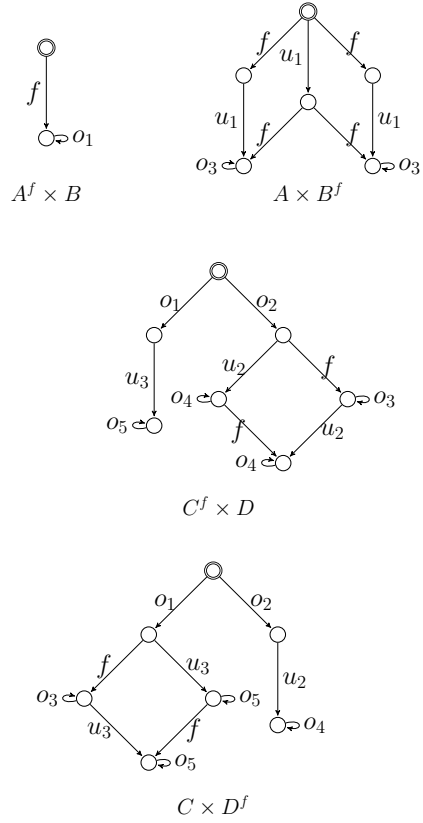

%Using Definition~\ref{def:red}, it is easy to prove that a trace of $G_i \times G_j$ does not belong to the traces of $G_i^f \times G_j$ if and only if one of its projections in $G_i$ contains an occurrence of the fault $f$.
%
%\begin{proposition} \label{proj_with_fault} % prop 3
%  $\forall \sigma \in \traces{G_i \times G_j} : \sigma \not \in \traces{G_i^f \times G_j} \iff  \exists \sigma_i : \proj{i}{\sigma} = \sigma_i \text{ and } f \in \sigma_i$
%\end{proposition}

The following result states necessary conditions for the diagnosability of the global system, i.e. the non diagnosability of $G_1^f \times G_2$ or $G_1 \times G_2^f$ implies the non diagnosability of $G_1 \times G_2$.

\begin{theorem} \label{the:1}
  Let $G_1$ and $G_2$ be two LTSs, then 
$$\diag{G_1 \times G_2} \ \ \Rightarrow\ \  \diag{G_1^f \times G_2} \land \diag{G_1 \times G_2^f}$$
\end{theorem}
%\vspace{0.3cc}
\begin{proof}
Lets assume that $\neg \diag{G_1^f \times G_2}$, then there exist two traces $\sigma, \alpha \in \traces{G_1^f \times G_2}$ and $f$ such that $obs(\sigma)=obs(\alpha)$ with $f \in \sigma$, but $f \not \in \alpha$. We know from Proposition~\ref{def:red} that every trace in $G_1^f \times G_2$ is a trace in $G_1 \times G_2$, so we have found two traces of the global system with the same observability, one containing a fault and the other one not. Therefore $(G_1 \times G_2)$ is non-diagnosable. An analogous analysis can be made if $\neg \diag{G_1 \times G_2^f}$.
\end{proof}

\begin{example}
  We see in Example~\ref{ex:reduced} that $C^f \times D$ is non diagnosable. Using Theorem~\ref{the:1} we can conclude that $C \times D$ is non diagnosable. This result is consistent with the diagnosability analysis made in Example~\ref{ex:prod}.
\end{example}
\vspace{1cc}
As explained above, the idea is to build a diagnosable component and to test that its interaction with another fault free component is also diagnosable. We can then decide the diagnosability of $G_1 \times G_2$ in term of the diagnosability of $G_1, G_2, G_1^f \times G_2$ and $G_1 \times G_2^f$.
\vspace{1cc}

\begin{theorem} \label{the:2}
  Let $G_1$ and $G_2$ be two LTSs, then
\begin{eqnarray*}
  \left.
  \begin{array}{llc} 
    &\diag{G_1} \land \diag{G_1 \times G_2^f} \\
    & \\
    & \diag{G_2} \land \diag{G_1^f \times G_2} \\
  \end{array}
  \right \}
  \Rightarrow \diag{G_1 \times G_2}
\end{eqnarray*}
\end{theorem}

\begin{proof}
  Let assume that we have a fault $f \in \Sigma_F$ and two traces $\sigma, \alpha \in \traces{G_1 \times G_2}$ with $f \in \sigma$ and $obs(\sigma)=obs(\alpha)$, we need to prove that $f \in \alpha$. Consider the following cases:
  \begin{enumerate}
    \item if $\sigma, \alpha \in \traces{G_i^f \times G_j}$ we can prove by $(G_i^f \times G_j)$'s diagnosability that $f \in \alpha$ and then $G_1 \times G_2$ is diagnosable,
%    \item[]
    \item if $\alpha \not \in \traces{G_i^f \times G_j}$, using the hypothesis that $\alpha \in \traces{G_i \times G_j}$, we can apply Proposition~\ref{def:red} and obtain that $\exists \alpha_i: \proj{i}{\alpha} = \alpha_i  \land f \in \alpha_i$. By Proposition~\ref{proj_inh_fault} we know that every fault belonging to a projection also belongs to the trace in the global system, then $f \in \alpha$ and $G_1 \times G_2$ is diagnosable,
%    \item[]
    \item if $\alpha \in \traces{G_i^f \times G_j}$ and $\sigma \not \in \traces{G_i^f \times G_j}$ we know by Proposition~\ref{def:redprod} that $\forall \alpha_i : \proj{i}{\alpha}=\alpha_i$ and $f \not \in \alpha_i$ and also that $\exists \sigma_i : \proj{i}{\sigma}=\sigma_i$ with $f \in \sigma_i$. As $obs(\sigma)=obs(\alpha)$ we have that $obs(\sigma_i)=obs(\alpha_i)$ by Proposition~\ref{proj_obs}. Finally as $G_i$ is diagnosable and $f \in \sigma_i$, the fault should belong to $\alpha_i$, leading to a contradiction. We can conclude that $G_1 \times G_2$ is diagnosable.
  \end{enumerate}
\end{proof}

\begin{example}
  From Example~\ref{ex:prod} and Example~\ref{ex:reduced} we know that $A,B, A^f \times B$ and $A \times B^f$ are diagnosable. If we apply Theorem~\ref{the:2} we can conclude that $A \times B$ is diagnosable, which is consistent with the analysis made in Example~\ref{ex:prod}.
\end{example}

%%%%%%%%%%%%%%%%
% section secGeneralizacion
%%%%%%%%%%%%%%%%
\subsection{Generalization}

Until now we only consider systems composed by only two components. However, real examples are usually more complex and are composed of several components. Therefore we need to generalize the previous results to global systems composed of $n$ different components running in parallel.

In order to generalize all our results, the associativity and commutativity property of synchronous product become essential. Note that in a general case the set of synchronizing actions is not necessarily the intersection of all their observable actions. Suppose that a system is composed by three components, $G_1, G_2$ and $G_3$, where two of them synchronize via an action $a$ that does not belong to a third component, i.e. $a \in \Sigma_o^1 \cap \Sigma_o^2$, but $a \not \in \Sigma_o^3$. We expect that $G_1$ and $G_2$ still synchronize in $a$. Fortunately, despite its apparent complications, the synchronous product is associative and commutative. The proof of such result can be found in previous work~\cite{gonza}. 

The following results generalized Theorems~\ref{the:1} and~\ref{the:2} respectively, giving necessary and sufficient conditions for the diagnosability of the global system.

\begin{theorem} \label{the:3}
  Let $G_1, G_2, \dots , G_n$ be $n$ components modeled by LTSs, then \\
  $\def\arraystretch{1.5}
  \begin{array}{c}
    \hspace{20mm} \diag{G_1 \times G_2 \times \dots \times G_n} \\
    \hspace{20mm} \Downarrow \\
    \hspace{20mm} \overbrace{\diag{G_1 \times G^f_2 \times \dots \times G^f_n}\ \land}\\
    \hspace{20mm} \diag{G_1^f \times G_2 \times \dots \times G^f_n}\ \land \\
    \hspace{25mm} \vdots \\
    \hspace{20mm} \diag{G^f_1 \times G_2^f \times \dots \times G_n} \ \ \ \ \\
  \end{array}$
\end{theorem}

\begin{theorem} \label{the:4}
  Let $G_1, G_2, \dots , G_n$ be $n$ components modeled by LTSs, then \\
  $\def\arraystretch{1.5}
  \begin{array}{c}
    \hspace{10mm} \diag{G_1} \land \diag{G_1 \times G^f_2 \times \dots \times G^f_n}\ \land \\
    \hspace{10mm} \diag{G_2} \land \diag{G_1^f \times G_2 \times \dots \times G^f_n}\ \land \\
    \hspace{15mm} \vdots \\
    \hspace{10mm} \underbrace{\diag{G_n} \land \diag{G^f_1 \times G_2^f \times \dots \times G_n} \ \ \ \ } \\
    \hspace{10mm} \Downarrow \\
    \hspace{10mm} \diag{G_1 \times G_2 \times \dots \times G_n} \\
  \end{array}$
\end{theorem}

Their proofs can be inferred directly from results that can be found in~\cite{gonza}. 

When the faults can occur in every component and $G_1^f \times G_2^f \times \dots \times G_n \neq G_1 \times G_2 \times \dots \times G_n$, our approach shows important advantages, however in the cases where $G_1^f \times G_2^f \times \dots \times G_n = G_1 \times G_2 \times \dots \times G_n$, the whole product is analyzed and the computation time of our method is equal to the classic one. Nevertheless, when a diagnosability analysis is performed it is because it is known that several faults can occur in different components and it is more likely that $G_1^f \times G_2^f \times \dots \times G_n$ is smaller than $G_1 \times G_2 \times \dots \times G_n$.

Moreover, the diagnosability analysis of each component and $G_1^f \times G_2^f \times \dots \times G_n$ can be tested in parallel, allowing parallel analysis of diagnosability.

%%%%%%%%%%%%%%%%
% section secDDES
%%%%%%%%%%%%%%%%
\section{The DADDY tool}\label{sec:daddy}

In the previous section we try to minimize the information that components needs to share to be able to decide the diagnosability property of the whole system. We now present our tool, called DADDY (from Distributed Analysis for distributed Discrete sYstems). DADDY implements the method presented above and the classic one (where the synchronous product is computed before the diagnosability analysis is performed). The tool is written in Python and has GNU GPL v3 license. It uses a standard format (.aut) for the description of each component and it also allows to see a graphical representation of the system. It can be downloaded from~\cite{daddy}.

The tool receives as inputs the components of the system. These inputs are assumed to be diagnosable, if not, an alert message is returned. If the specifications, meaning the non faulty components, are not given, systems $G^f_j$, for $j \not = i$, are computed following Algorithm~\ref{algo}. Hence $G^f_j$ is synchronized with $G_i$, and its diagnosability is  checked using the twin plant method from~\cite{twin}. Also, time $t_i$ of such computation is registered.

As soon as it is known that a component interacting with fault free versions of the other ones is non diagnosable, applying Theorem~\ref{the:3}, a non diagnosable verdict is returned. Moreover, using the fact that it is a distributed computation, when we find a non diagnosable component, the computation of all others components can be stopped. So, the resulting time of such computation is $min(t_i)$ with $1 \leq i \leq n$.% for non-diagnosable components. 

On the other hand, if every component interacting with the fault free version of the other ones is diagnosable, using the assumption that every $G_i$ is diagnosable by its own, we can conclude that $G_1 \times \dots \times G_n$ is diagnosable applying Theorem~\ref{the:4}. In this case, the diagnosability of every component is computed (in parallel) and the required time is $max(t_i)$ with $1 \leq i \leq n$.

We can see in table from Figure~\ref{fig:table} that the diagnosability analysis results obtained by DADDY are consistent with the ones presented in our previous examples. We can also see that our method can be almost ten times faster than the classical one. If we consider systems $n_1, n_2, n_3$ from exaples/sample5 in~\cite{daddy}, a non diagnosable result is obtained (as $n_1^f \times n_2 \times n_3^f$ is not diagnosable) in 0.16974902153 seconds with our method while the classical one does not reach a result after more than 24 hours. This shows an important improvement with respect to the classical method when the number of components grows.

%In the special cases of $B \times C$ and $B \times D$, the synchronous product after removing the fault remains the same as the global system. Therefore the performance of both methods is almost the same.

%%%%%%%%%%%%%%%%
% section sec:CFW
%%%%%%%%%%%%%%%%
\section{Conclusions and Future Work}\label{sec:CFW}

We have presented a new framework for the distributed diagnosability analysis of concurrent systems. We remove the assumption that a fault can only occur in a single component (which is usually made in distributed systems) and allow to analyze more general systems. The method presented in this paper parallelized the analysis leading, in general, to an important reduction in the computing time. The theoretical results are illustrated by several examples and supported by experimental results obtained with the DADDY tool.

%We plan to continue in this direction and try to keep reducing the system in order to obtain minimal subsystems from which we can infer the diagnosability of the original global system.
We plan to continue trying to keep reducing the system in order to obtain minimal components from which we can infer the diagnosability of the original global system. In addition, we intend to relax the assumption that the communicating (synchronizing) events are observable.

Furthermore, even if the framework presented in this paper allows the distribution of the analysis, the formalism to model the systems is still sequential (product of LTSs) and can suffer of state space explosion making the twin plant method to check its diagnosability still prohibitive. We are working to extend such analysis to concurrent models such as Petri Nets. \\

\begin{figure}
\begin{center}
\footnotesize
%\scriptsize
\begin{tabular}{|c|c|c|c|}
\hline
System & Diagnosable & Our method & Classic method \\
\hline
& & 0.0027251243 & 0.024051904 \\
& & 0.0028400421 & 0.023932933 \\
$A \times B$ & yes & 0.0028848648 & 0.024003028 \\
& & 0.0029160976 & 0.025793075 \\
& & 0.0032229423 & 0.023809194 \\
\hline
%& & 0.0024521350 & 0.018491983 \\
%& & 0.0026359558 & 0.023923158 \\
%$A \times C$ & yes & 0.0028109550 & 0.020237922 \\
%& & 0.0023770332 & 0.014995098 \\
%& & 0.0030820369 & 0.020065784 \\
%\hline
%& & 0.0023820400 & 0.024647951 \\
%& & 0.0023059844 & 0.023785114 \\
%$A \times D$ & no & 0.0023379325 & 0.028781890 \\
%& & 0.0024361610 & 0.028961181 \\
%& & 0.0023980140 & 0.023897171 \\
%\hline
%& & 0.0012741088 & 0.001409053 \\
%& & 0.0008151531 & 0.001728057 \\
%$B \times C$ & yes & 0.0011539459 & 0.001488924 \\
%& & 0.0012412071 & 0.001722812 \\
%& & 0.0008060932 & 0.001607179 \\
%\hline
%& & 0.0011439323 & 0.001439094 \\
%& & 0.0012350082 & 0.001441955 \\
%$B \times D$ & yes & 0.0011658668 & 0.001447200 \\
%& & 0.0010938644 & 0.002552032 \\
%& & 0.0012590885 & 0.001463890 \\
%\hline
& & 0.0041198730 & 0.015272855 \\
& & 0.0040440559 & 0.015629053 \\
$C \times D$ & no & 0.0042178630 & 0.015436887 \\
& & 0.0040760040 & 0.009753942 \\
& & 0.0047080516 & 0.015598058 \\
\hline
\end{tabular} 
\end{center}
  \caption{Diagnosis results in seconds unit}
  \label{fig:table}
\end{figure}

\bibliographystyle{dx-2013}
\bibliography{diagnosis}
\end{document}